\documentclass[onecolumn,12pt]{IEEEtran}

\usepackage{graphicx}    
\usepackage{hyperref}    
\usepackage{setspace}    
\usepackage{verbatim}    

\usepackage{amsmath}
\usepackage{amsthm}
\usepackage{amssymb} 
\usepackage{color}

\usepackage{tikz}

\usetikzlibrary{
    arrows.meta,
    positioning,
    calc
}

\newcommand{\suchthat}{\; | \;}

 \newcommand{\be}{\begin{equation}}
\newcommand{\ee}{\end{equation}}

\newtheorem{Problem} {Problem}

\newtheorem{Theorem}{Theorem}

\newcommand{\R}{\mathbb R}

\def\TK{\textcolor{magenta}}
\definecolor{darkgreen}{rgb}{0,0.6,0}

\newcommand{\significancestatement}[1]{%
  \begin{center}
    \parbox{0.95\textwidth}{
      \small
      \textbf{ \begin{center}Significance Statement\end{center}} #1
    }
  \end{center}
}

\title{A universal multi-turnpike principle for optimal allocation of translational resources}

\author{Ram Massas, Thomas Kriecherbauer, Lars Gr\"une, Tamir Tuller, and  Michael Margaliot\thanks{RM  and MM   are with the School of ECE, Tel Aviv University, Tel Aviv, Israel~69978. TK and LG are with Mathematical Institut, University of Bayreuth, 95440 Bayreuth, Germany. TT is with the  Faculty of Engineering
and the Edmond J. Safra Center for Bioinformatics,
Tel Aviv University, Israel 69978.
Correspondence:   michaelm@tauex.tau.ac.il    }}

\begin{document}

\maketitle

 \begin{abstract}
mRNA translation   in the cell requires efficient allocation of shared and limited resources including  free  ribosomes, tRNA molecules, and initiation factors across multiple transcripts. 
 Using a network of  dynamic  mathematical models for ribosome flow along the mRNA, we pose the problem of maximizing the total steady-state protein production rate in the cell
 under a shared and
 limited total budget for all translation rates in all the transcripts. 
 We prove  that 
 the optimal solution of this  resource allocation problem 
 admits a multi-turnpike structure:   in each mRNA, the transition rates are high and  nearly uniform along the bulk of the coding region, with lower and varying  rates near   the boundaries of the~mRNA.    Our results are based on 
 the emergence of hierarchical optimality: 
  regardless of how resources are allocated among genes,
    every transcript should internally organize itself in essentially the same way.   This  suggests that to  
  optimize the overall production rate it is sufficient  to regulate the initiation and termination regions  in each transcript.  Remarkably, this universal turnpike structure holds for any number of transcripts, arbitrary transcript lengths, and various optimization criteria.
   This agrees with  observed  conserved translational phenomena, such as codon ramps and initiation-dominated regulation. 
      Our findings may also  provide guidelines for the rational design of intracellular circuits operating under translational control.
 \end{abstract}

\significancestatement{
 Cells face the challenge of allocating limited translational resources among competing genes. We use a mathematical model for large-scale translation in the cell to show that maximizing protein production under shared resource constraints   gives rise to a multi-turnpike structure: transition rates in all transcripts are nearly uniform across most of the coding region, while regulation is concentrated near transcript boundaries. This result provides a systems-level explanation for conserved translational phenomena, such as codon ramps and initiation-dominated regulation. Our results suggest that the ubiquitous codon ramps observed across organisms may not reflect transcript-specific adaptations alone, but rather a universal consequence of optimal allocation of shared translational resources. Our results also yield practical design principles for engineering translational control in synthetic biology.
}
 

\section{Introduction}

Cells must continuously allocate limited gene-expression resources, including ribosomes, tRNAs, and translation factors, among many  competing transcripts. For example, a   mammalian cell includes approximately $10^5-10^6$ mRNA molecules  \cite{num_mrna_2014} that indirectly compete for shared resources like initiation factors, free  ribosomes, and  tRNA molecules. In principle, optimal resource allocation 
  requires   (1)~sensing various  internal and external conditions;
(2)~determining an appropriate allocation policy;
and (3)~implementing it through complex regulatory mechanisms acting at multiple levels of gene expression. Such resource-allocation challenges are central to cellular growth, adaptation, and survival~\cite{Zeng2021CellularResourceAllocation,TEUSINK2026103391,total_mrna_PNAS2024}. Understanding large-scale resource allocation in the cell 
is also important for 
the rational design     of interconnected translational modules in synthetic biology~\cite{resource-aware2024}.

Revealing global principles related to translation optimization is an important goal that requires 
averaging over a large number of coding regions~\cite{NAR2015,TULLER2010344}, 
  as patterns observed in single genes 
  tend to be blurred by  additional overlapping codes. 
In many organisms,  the mean values of  bio-physical 
features affecting ribosomal speed during translation (e.g., adaptation to the tRNA pool, local mRNA folding, local amino acid charge) display 
a ``turnpike structure'':   
decoding rates in the middle of the coding region tend  to be uniform, whereas the rates in the initiation and termination regions vary, sometimes referred to as ``codon ramps''~\cite{Peeri2020,TULLER2010344,Bahiri2021,Jia_2026,TULLER2011}.
Slower rates along  the ``initiation ramp'' improve translation efficiency by preventing the formation of   
ribosomal ``traffic jams'' along the transcript~\cite{TULLER2010344}.
Frumkin et al.~\cite{pilpel2017}
found that three gene architecture factors: codon decoding   time, 
mRNA structure, and affinity to the anti-Shine
Dalgarno motif contribute to slowing down ribosomes near the 5' end of the transcript (that is, the beginning of the RNA molecule).

Since mRNA translation is the most energy consuming process in the cell and  plays a pivotal role in the control of gene expression~\cite{Roux2012},
it was suggested that these   patterns of decoding rates   
were shaped by evolution so as
to optimize translation efficiency~\cite{Peeri2020,TULLER2010344,Bahiri2021,TULLER2011}.

Here, we    analyze    the problem of how to   distribute  the limited shared resources  across multiple mRNA molecules  in order to optimize the overall protein production in the cell. 
We use a nonlinear dynamical 
model  for   the sequential 
progression   of ribosomes   along the mRNA   called  the ribosome flow model (RFM)~\cite{reuveni2011genome}.
This is a phenomenological model for the unidirectional flow of ``particles'' along an ordered set of~$n$ sites. 
This model  includes a  soft simple exclusion principle,
representing 
the fact that multiple  ribosomes can scan the same mRNA in parallel, 
but cannot`overtake one  another. 
In the RFM,  groups of consecutive codons along the transcript are coarse-grained into~$n$ ordered   sites, and   local biophysical features such as codon usage, cognate tRNA abundance, and mRNA structure 
are encapsulated into a  positive transition rate~$\lambda_i$ for  each site~$i$. 
In particular, the rate~$\lambda_0$ controls the entry rate into the first site, and~$\lambda_n$
controls the exit rate from the last site (i.e., site~$n$).

The RFM with~$n$ sites 
admits a unique steady  state  of
ribosome densities and hence a unique steady state protein
production rate~$R$ 
that depends on the $n+1$ transition rates (but not on the initial ribosome densities). 
The~RFM   has been   used to analyze translation efficiency and   regulation along  a single mRNA molecule~\cite{rfm_sense, rfm_max,  margaliot2012stability,RFM_NEGATIVE_FEEDBACK}.
Competition for ribosomes has been modeled and analyzed 
using a 
network of RFMs interconnected via a  shared 
pool of free  ribosomes~\cite{Raveh2016}. Ribosomes exiting an mRNA return to the pool  and are then redistributed among all transcripts. These studies showed that competition for the finite ribosome pool induces indirect coupling and global resource allocation effects, leading  to nontrivial trade-offs between translation efficiency of individual mRNAs and overall protein production in the network~\cite{Raveh2016,aditi_networks,fierce_compete}.

 A recent paper~\cite{kaminer2026turnpikepropertyeigenvalueoptimization}
considered the problem of maximizing the steady state protein production rate~$R$
  subject to an upper bound on the sum of all the transition rates along the mRNA.  This constraint  reflects a cellular resource allocation principle: since translational resources are limited,   increasing the elongation rate at one position necessarily reduces what can be allocated elsewhere.
    It was shown in~\cite{kaminer2026turnpikepropertyeigenvalueoptimization}
    that   
a unique optimal solution exists and  admits a turnpike structure:     
the ordered list of optimal transition rates      includes three parts with the first and third part relatively short, and the values in the middle part are all approximately equal.

Here, we consider a system comprising $m$ RFMs with potentially different lengths and study the optimal allocation of a fixed total budget of transition rates across all~RFMs. The goal is to maximize a ``utility function''~$F$ that models the overall steady-state protein production in the cell. Let $R_i$
 denote the steady-state production rate of the $i$th RFM. The objective function~$ F(R_1,\dots,R_m)$
 is assumed to be   increasing in each argument, and may otherwise be quite general. A simple example is the total production rate, 
 $F=R_1+\dots+R_m$. 
This formulation gives rise to a hierarchical resource-allocation
  problem. At the local level, increasing a transition rate at a particular site within an RFM necessarily reduces the budget available for the other transition rates in the same transcript. At the global level, allocating more resources to one transcript reduces the budget available to all other transcripts in the system. Thus, every increase in a transition rate entails both an intra-transcript trade-off and an inter-transcript trade-off. This captures the biological reality that translational resources are shared and limited, and must be distributed both along individual mRNAs and across multiple genes competing for expression.

Our main finding   is that an optimal solution for the entire network  
of~RFMs  exists and admits a \emph{multi-turnpike property}: the list of ordered transition rates in \emph{each} RFM admits a turnpike structure. 
  This result has a natural biological interpretation. Uniform elongation rates reduce the formation of local bottlenecks, enabling a smooth flow of ribosomes and efficient use of resources. In contrast, the boundary regions play distinct regulatory roles. 
  
These results provide a system-level explanation for widely observed conserved features of gene sequences and translation regulation.  It is well-known that 
  initiation is often rate-limiting and highly regulated, 
  and there is a growing interest in translational control near  stop codons~\cite{Jia_2026,TEODOROWICZ2026169765}.  
Termination pausing is usually related to translational readthrough  and  ribosome recycling, but our analysis suggests that it may also increase protein translation efficiency.

\section{The local optimization problem in a single transcript}

Consider a single RFM of length~$n>1$
with non-negative
transition rates~$\lambda=\begin{bmatrix} \lambda_0&\lambda_1&\dots&\lambda_n\end{bmatrix}^\top$,  and let~$R=R(\lambda)$
denote the 
steady state production rate (more details on the~RFM are given in the Supplementary Material).
Pose  the problem of optimizing~$R(\lambda)$
 subject to the constraint
 \[
 \sum_{i=0}^n\lambda_i\leq n+1.
 \]
 The total bound on the transition rates implies that if sites get higher transition rates then others must get less. 
Thus, this optimization problem       balances resource allocation  across  the   different sites along a single transcript.
The total bound on the transition rates scales with the length 
of the RFM, as an RFM with~$n$ sites has~$n+1$ transition rates. Note that the choice $\lambda_i = 1 $ for all~$i$  satisfies the constraints, but this is not an
optimal solution.

 A unique  optimal solution~$\bar\lambda$ exists and is known~\cite{min_spring}  to be symmetric with respect to the middle of the chain, that is,  
\be\label{eq:symm}
\bar \lambda_i=\bar\lambda_{n-i},
\quad i=0,\dots,n,
\ee
and strictly increasing up to the middle of the chain, that is, 
\be\label{eq:lam_mono}
\bar\lambda_0<\bar\lambda_1<\dots<\bar\lambda_{\lfloor n/2 \rfloor}.
\ee
Furthermore, the solution~$\bar \lambda
 $ admits a turnpike structure~\cite{kaminer2026turnpikepropertyeigenvalueoptimization}: for any~$n\geq 36$, there exists a value~$\bar\sigma=\bar\sigma(n)>0$ such that the optimal transition rates satisfy 
\be\label{eq:36}
0<\frac{4}{\bar\sigma ^2}-\bar \lambda_i<\frac1{2^i} \text{ for all } i=0,1,\dots,\lfloor n/2\rfloor.
\ee
Combining this with~\eqref{eq:symm} shows that in the bulk the optimal rates are nearly uniform and   approximately equal   to the value~$4/\bar \sigma^2$. 
     \begin{figure}[t]
\begin{center}
\includegraphics[scale=0.8]{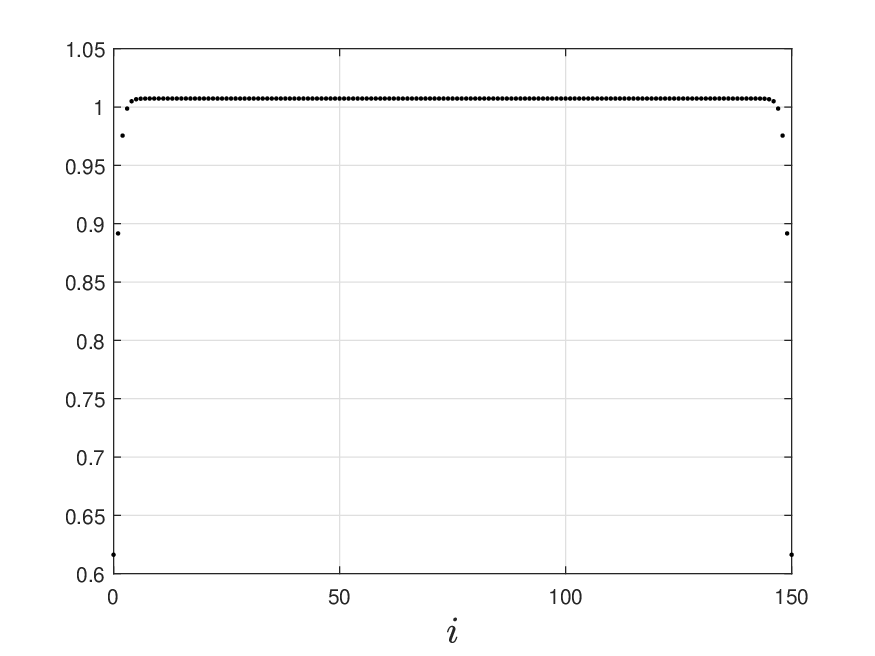}
  \caption{Optimal transition rates~$\bar\lambda_i$
  as a function of~$i$ for~$n=150$. In the bulk all the optimal rates are approximately equal to the value~$\frac{4}{\bar\sigma^2}\approx 1.0073$. }
  \label{fig:opt150}
\end{center}
\end{figure}
Fig.~\ref{fig:opt150} shows the optimal transition rates for the case of an~RFM with~$n=150$ sites. 

The structure of the model implies that if we change the constraint
\[ \sum_{i=0}^n\lambda_i\leq n+1 \quad \mbox{to} \quad \sum_{i=0}^n\lambda_i\leq \rho(n+1) \]
for some $\rho>0$, then the optimal solution changes from $\bar\lambda$ to $\rho\bar\lambda$, i.e., the optimal rates scale linearly with the constraints. 

\section{Multi-turnpike structure in large-scale translation in the cell}\label{sec:multi}
Analyzing a 
single~RFM provides valuable insights into the translation dynamics along an individual transcript  such as the effects of initiation, elongation, and local bottlenecks.
However, cellular protein synthesis is   a system-level process involving many mRNAs that compete for shared resources, including    free ribosomes and initiation factors. Consequently, to capture large-scale translation in the cell, it is necessary to consider networks of coupled RFMs, where interactions between transcripts give rise to collective effects such as resource allocation, competition, and global regulation.  

Here,
we pose a   global  optimization problem for~$m$   RFMs of different lengths, coupled
through a constraint on all the translation rates in all the RFMs, 
and show that the optimal solution of the global problem  satisfies a \emph{multi-turnpike structure}. 
To simplify the presentation and the 
notation, we consider the case~$m=2$, that is, 
 two~RFMs. In Appendix A, we describe the extension of the model and the results to the case of a general~$m$.

 Consider~$m=2$ RFMs, 
 the first  with length~$n_1$ and rates~$\lambda_0,\dots,\lambda_{n_1}$, and the second with length~$n_2$ and rates~$\eta_0,\dots,\eta_{n_2}$.
Let~$R_1=R_1(\lambda)$ [$R_2=R_2(\eta)$]
denote the steady state production rate in the first [second]~RFM. We model the total production rate in the cell
using a continuous utility  function~$F(R_1(\lambda),R_2(\eta))$ that is   increasing in~$R_1(\lambda)$ and in~$R_2(\eta)$,  and satisfies an additional technical 
condition described in the Supplementary Material. 
This general formulation encompasses a wide range of biologically meaningful objective functions.  Examples include:
\begin{itemize}
\item 
$F_1(R_1(\lambda),R_2(\eta)):=w_1 R_1(\lambda)+w_2 R_2(\eta)$,
where~$w_1,w_2$ are positive weights. For example, if   the transcripts 
encode the same protein $P$, and~$w_1=w_2=1$ then this function is the total production rate of~$P$ in the cell.
\item $F_2(R_1(\lambda),R_2(\eta)) : =w_1 \ln(1+R_1(\lambda))+w_2\ln(1+R_2(\eta))$, 
where~$w_1,w_2$ are positive weights. 
This function is similar to $F_1$,
but with the~$\ln(\cdot)$ function corresponding  to a saturation effect.
\item 
$F_3(R_1(\lambda),R_2(\eta))
:=R_1(\lambda) R_2(\eta)$. This is   biologically  relevant when 
the two transcripts encode two different proteins~$P_1$ and~$P_2$ that 
bind to form an active complex.
Then mass-action kinetics gives that the rate of complex formation
is proportional to~$R_1(\lambda) R_2(\eta)$. 
\item $F_4(R_1(\lambda),R_2(\eta)) : =\min\{R_1(\lambda),R_2(\eta)\}$. This is meaningful when the two proteins are used in an upstream process that is limited by the scarcer component.
\end{itemize}

We pose  the following optimization problem.
\begin{Problem}\label{prob:max_global}
Maximize~$F(R_1(\lambda),R_2(\eta))$
    subject to the constraints:
\be\label{eq:constr_multi}
\lambda_i\geq  0, \eta_j\geq 0 \text{ for all } i,j,  \text{ and}  
\sum_{i=0}^{n_1}\lambda_i+
\sum_{j=0}^{n_2}\eta_j \leq n_1+1+n_2+1.
\ee
\end{Problem}

Thus, 
 the two independent RFMs are  only
 coupled
  via the constraint on the total ``budget''~$n_1+n_2+2$ for all the transition rates in all the transcripts. If more of the budget is allocated to the first transcript, then less is left to the second transcript. This encapsulates  for  example   the case where the total budget corresponds to free ribosomes, or tRNA molecules,
  that are a shared and limited resource. 
Problem~\ref{prob:max_global} thus     balances resource allocation both locally, across sites within each mRNA, and globally, across different~mRNAs.

 Note that the choice~$\lambda_i=1$ for all~$i$,
  and~$\eta_j=1$ for all~$j$ satisfies the constraints in~\eqref{eq:constr_multi}, but this is not an optimal solution.

Since the function~$F$ is continuous and the constraint~\eqref{eq:constr_multi} defines a closed and bounded feasible set, Problem~\ref{prob:max_global} admits at least one optimal solution~$(\bar\lambda,\bar\eta)$.

We can now state our main result revealing a universal turnpike principle for translation. 
\begin{Theorem}
\label{thm:main}
Let~$(\bar\lambda,\bar\eta)$
be an optimal solution of Problem~\ref{prob:max_global}. Denote
\begin{align*}
    \bar s_1&:=\sum_{i=0}^{n_1}\bar\lambda_i,\\
    \bar s_2&:=\sum_{j=0}^{n_2}\bar\eta_j. 
\end{align*}
Then~$\bar \lambda$ is the solution of the problem~$\max R(\lambda)$ under the constraint~$\lambda_i\geq 0$  for all~$i$ and~$\sum_{i=0}^{n_1}\lambda_i=\bar s_1$, and~$\bar \eta$ is the solution of the problem~$\max R(\eta)$ under the constraint~$\eta_j\geq 0$ for all~$j$ 
and~$\sum_{j=0}^{n_2}\eta_j=\bar s_2$. 
In particular, both~$\bar \lambda$ and~$\bar \eta$ are turnpike solutions. 
\end{Theorem}

 In other words, 
 regardless of how   resources are allocated among genes,
    every transcript should internally organize itself in essentially the same way. The global optimization problem yields a turnpike structure in each mRNA, that is, the transition 
rates in each transcript admit  a spatially homogeneous profile in the bulk, with deviations confined to a small boundary layer.

Figure~\ref{fig:5040}
depicts the optimal solution~$(\bar\lambda,\bar \eta)$
of Problem~\ref{prob:max_global} 
when~$n_1=50$, $n_2=40$, and~$F(R_1(\lambda),R_2(\eta))=R_1(\lambda)R_2(\eta)$. It may be seen that indeed  both~$\bar \lambda$ and~$\bar \eta$ are turnpike solutions.

\begin{figure}[t]
\begin{center}
\includegraphics[scale=0.8]{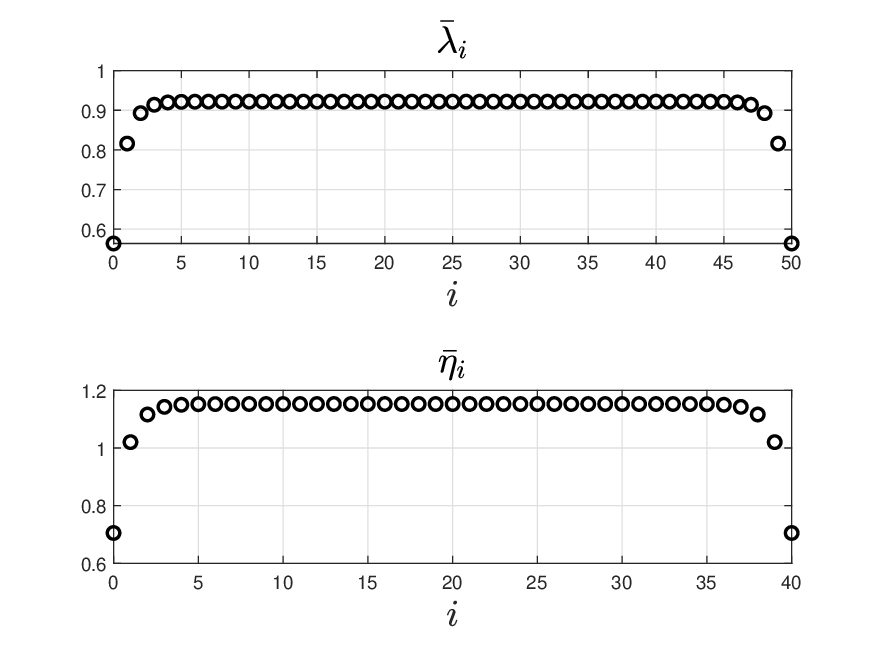}
  \caption{Optimal solution when~$n_1=50$, $n_2=40$, and~$F(\lambda,\eta)=R_1(\lambda)R_2(\eta)$. Top: optimal~$\bar\lambda_i$s  as a function of~$i$. Bottom: optimal~$\bar\eta_i$s  as a function of~$i$.  Here $R_1(\bar\lambda) =0.2304$,
  $
R_2 (\bar\eta)= 0.2882$, 
$\sum_{i=0}^{50} \bar\lambda_i =46$, and~$\sum_{i=0}^{40} \bar\eta_i =46$. }
  \label{fig:5040}
\end{center}
\end{figure}

The proof of the main result, given in the Supplementary Material,
relies on  a decomposition principle: although the optimization problem is global and highly coupled through the  shared resource constraint, every transcript must itself solve the single-transcript optimal allocation problem (see Figure~\ref{fig:hier}).

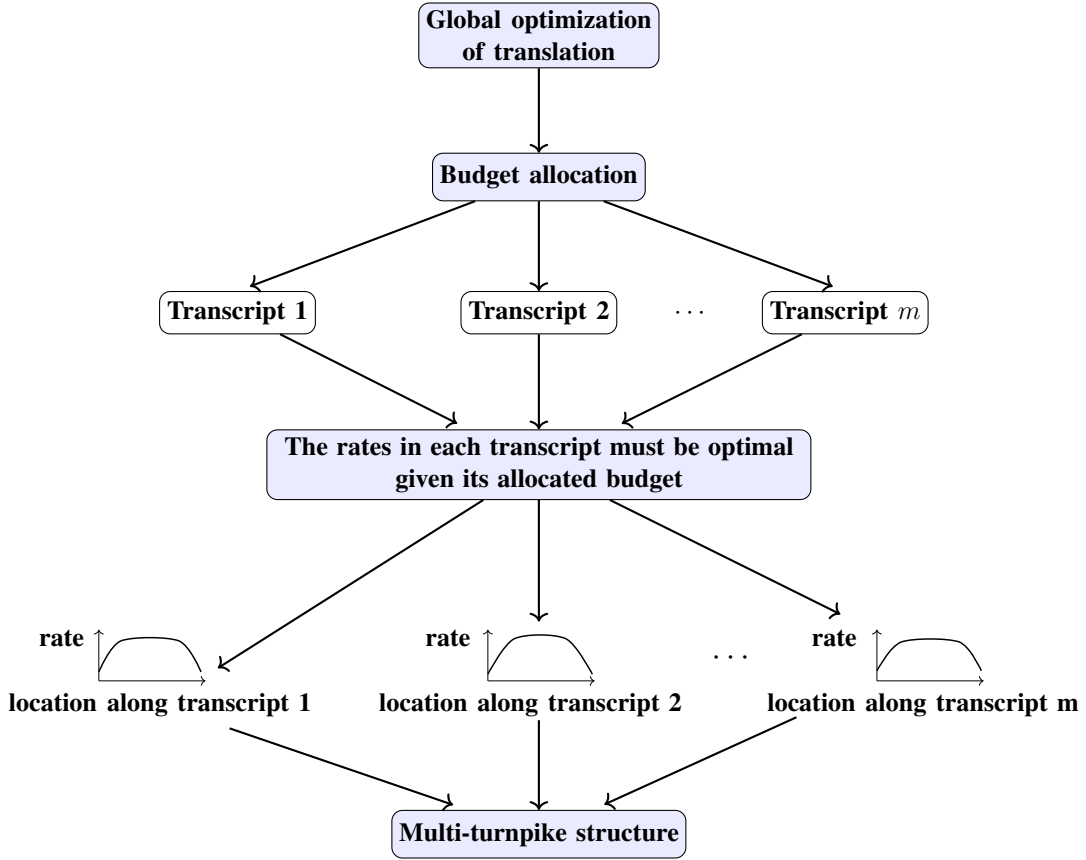
\begin{figure}[t]
\centering
\begin{tikzpicture}[  
scale=0.7,
    transform shape,
node distance=1.6cm,
every node/.style={font=\large},
box/.style={
draw,
rounded corners,
align=center,
minimum width=4cm,
minimum height=0.9cm,
fill=blue!8},
smallbox/.style={
draw,
rounded corners,
align=center,
minimum width=2.6cm,
minimum height=0.8cm,
 },
]


\node[box] (global)
{\bf Global optimization\\ \bf of translation};

\node[box,below=of global] (budget)
{\bf Budget allocation};

\draw[->,thick] (global)--(budget);


\node[smallbox,below left=1.7cm and 2.2cm of budget] (t1)
{\bf Transcript 1};

\node[smallbox,below=1.7cm of budget] (t2)
{\bf Transcript 2};


\node[smallbox,below right=1.7cm and 2.2cm of budget] (t3)
{\bf Transcript $m$};

\node at ($(t2)!0.5!(t3)$) {$\cdots$};

\draw[->,thick,shorten >=5pt] (budget)--(t1.north);
\draw[->,thick ] (budget)--(t2.north);
\draw[->,thick,shorten >=5pt] (budget)--(t3.north);


\node[box,below=1.8cm of t2,text width=10cm] (local)
{\bf The rates in each  transcript must be optimal\\
given its allocated budget};

\draw[->,thick,shorten >=5pt] (t1)--(local);
\draw[-> ,thick] (t2)--(local);
\draw[->,thick,shorten >=5pt] (t3)--(local);


\node[below left=2.3cm and 1.0cm of local] (p1)
{
\begin{tikzpicture}[scale=.42]
\draw[->] (0,0)--(4.8,0);  
\draw[->] (0,0)--(0,2.3);
\draw[thick]
plot[smooth]
coordinates{
(0,0.4)
(1,1.8)
(3.6,1.8)
(4.6,0.4)};
\end{tikzpicture}
};

\node[below left=3.5cm and  -1cm of local] (p1)
{
{\bf location along transcript 1} 
};

\node[below=2.3cm of local] (p2)
{
\begin{tikzpicture}[scale=.42]
\draw[->] (0,0)--(4.8,0);
\draw[->] (0,0)--(0,2.3);
\draw[thick]
plot[smooth]
coordinates{
(0,0.3)
(1.2,1.9)
(3.5,1.9)
(4.7,0.3)};
\end{tikzpicture}
};

\node[below right=2.3cm and 1cm of local] (p3)
{
\begin{tikzpicture}[scale=.42]
\draw[->] (0,0)--(4.8,0);
\draw[->] (0,0)--(0,2.3);
\draw[thick]
plot[smooth]
coordinates{
(0,0.45)
(1.1,1.75)
(3.7,1.75)
(4.7,0.45)};
\end{tikzpicture}
};

\node at ($(p2)!0.5!(p3)$) {\bf \Large$\cdots$};

\draw[->,thick,shorten >=12pt] (local)--(p1);
\draw[->,thick] (local)--(p2);
\draw[->,thick,shorten >=10pt] (local)--(p3);


\node[box,below=2.2cm of p2]
(final)
{\bf Multi-turnpike structure};

\draw[->,thick,shorten >=6pt,shorten <=5pt] (p1)--(final);

\draw[->,thick,shorten <=10pt] (p2)--(final);

\draw[->,thick,shorten >=5pt,shorten <=28pt] (p3)--(final);

\node[below left=2.3cm and 3.3cm of local] (p1)
{
{\bf rate} 
};

\node[below left=2.3cm and -4.0cm of local] (p2)
{
{\bf rate} 
};
\node[below left=3.5cm and  -8cm of local] (p1)
{
{\bf location along transcript 2} 
};

\node[below left=2.3cm and  -11.3cm of local] (p3)
{
{\bf rate} 
};

\node[below left=3.5cm and  -15.5cm of local] (p2)
{
{\bf location along  transcript m} 
};
\end{tikzpicture}
\caption{Global optimization divides  the total 
available translational budget among the~$m $ transcripts. Each transcript must have optimal rates  given its allocated budget, leading to a turnpike profile of transition rates. Together these local solutions give rise to a multi-turnpike structure. }\label{fig:hier}
\end{figure}

Appendix A shows that the multi-turnpike structure of the optimal solution  holds also when the number of RFMs is~$m>2$.

 
\section{Optimal allocation of the total translational budget}

Theorem~\ref{thm:main} shows that the optimal solution for a network of RFMs exhibits a \emph{multi-turnpike structure}, that is, the optimal transition rates in each individual RFM exhibit a turnpike structure. How is the total budget   distributed among 
the~RFMs? For example, can the optimal solution allocate the entire budget to a single~RFM? Conversely, can it allocate exactly half of the budget to each of two~RFMs? The answers depend on both the lengths of the~RFMs and the specific choice of the utility function $F(R_1(\lambda),R_2(\eta))$.
A detailed analysis of these questions is provided in the Supplementary Material; here we highlight several of the main results.

\subsection{Optimizing the sum of the production rates}
Consider the function~$F_1(R_1,R_2)=w_1R_1+w_2R_2$,
and assume that~$w_1=w_2$. 
If the RFMs have the same length, that is,~$n_1=n_2$, then the total budget is~$N:=n_1+1+n_2+1=2n_1 +2 $. 
A solution that allocates a budget of~$q N$ to the first RFM and~$(1-q)N$ to the second, for \emph{any}~$q\in[0,1]$, is optimal, as long as
any RFM maximizes its steady state  production rate for the given budget. This is reasonable as the two~RFMs are identical  and their optimal rates scale linearly with $q$ and $1-q$, respectively. 

 If~$w_1=w_2$ and the RFMs have different lengths we have strong evidence that the optimal allocation of the budget is always to allocate all the budget~$n_1+n_2+2$ to the shorter~RFM (see Figure~\ref{fig:5045}). We have a mathematical proof for this statement if the longer~RFM has at least two more sites than the shorter one (see the Supplementary Material). 

\begin{figure}[t]
\begin{center}
\includegraphics[scale=0.8]{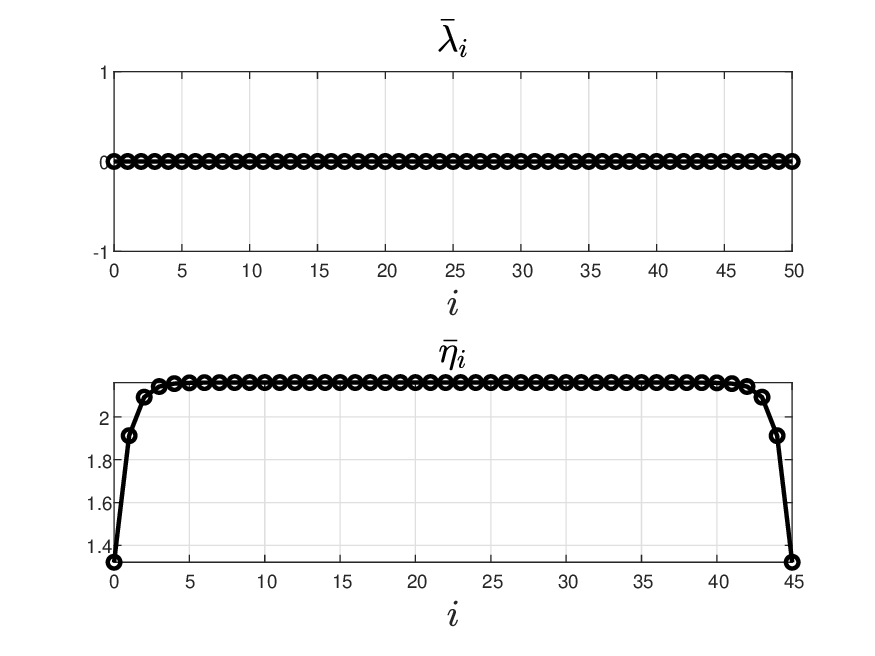}
  \caption{Optimal solution when~$n_1=50$, $n_2=45$, and~$F(R_1,R_2)=R_1+R_2$. Top: optimal~$\bar\lambda_i$s  as a function of~$i$. Bottom: optimal~$\bar\eta_i$s  as a function of~$i$.  Here~$R_1(\bar \lambda)= 0, R_2(\bar \eta) 
=  0.5400 $.  }
  \label{fig:5045}
\end{center}
\end{figure}

\subsection{Optimizing the product
of the production rates}
Consider the utility function~$F_3(R_1,R_2)=R_1R_2$.
In contrast to the objective $R_1+R_2$, which may favor allocating all the  resources to one~RFM at the expense of the other, the product objective strongly penalizes unbalanced allocations: if either
$R_1$ or~$R_2$
 becomes too small, then so does~$F_3$.
 Consequently, the optimal strategy is always to split the total budget equally, allocating one half to each RFM (see Figure~\ref{fig:5040}).

\subsection{Optimizing  the minimum of the  production rates}
Consider the function~$F_4(R_1(\lambda),R_2(\eta)) : =\min\{R_1(\lambda),R_2(\eta)\}$.
Let~$N:=n_1+n_2+2$ denote the total budget. There exists a unique~$\bar q \in[0,1]$ such that the  
optimal solution is   to allocate a proportion~$\bar q $ of the total budget
to the first RFM, and~$(1-\bar q ) $ to the second~RFM. The value~$\bar q$ depends on the lengths~$n_1$ and~$n_2$ of the RFMs, and is equal to~$1/2$ when~$n_1=n_2$. 
  
Summarizing, we see that different cellular objectives induce different translational allocation strategies between the different transcripts.

Figure~\ref{fig:plit10} depicts  the optimal rates in each~RFM in a network of 10 RFM with lengths
\be\label{eq:leng}
300, \, 325,\, 350,\, 375,\, 400,\, 425,\, 450,\, 465,\, 485,\, 500 , 
\ee
and utility function~$F(R_1,\dots,R_{10})=\min\{R_1,\dots, R_{10}\}$.   The total budget is~$301+326+\dots+501=4085$. In the optimal solution,~$\bar R_k=0.25076$ for all~$k$. The reason for this property of the optimal solution can be explained
as follows: if there were different outputs for different transcripts then one could reallocate some of the resources from the more productive RFMs to the less productive RFMs and obtain a larger value for the minimum~$F$. In addition, it
may be seen that
despite the different transcript lengths and different allocated budgets, the optimal rates in every transcript
admit    a  turnpike structure.

 \begin{figure}[t]
\begin{center}
\includegraphics[scale=.55]{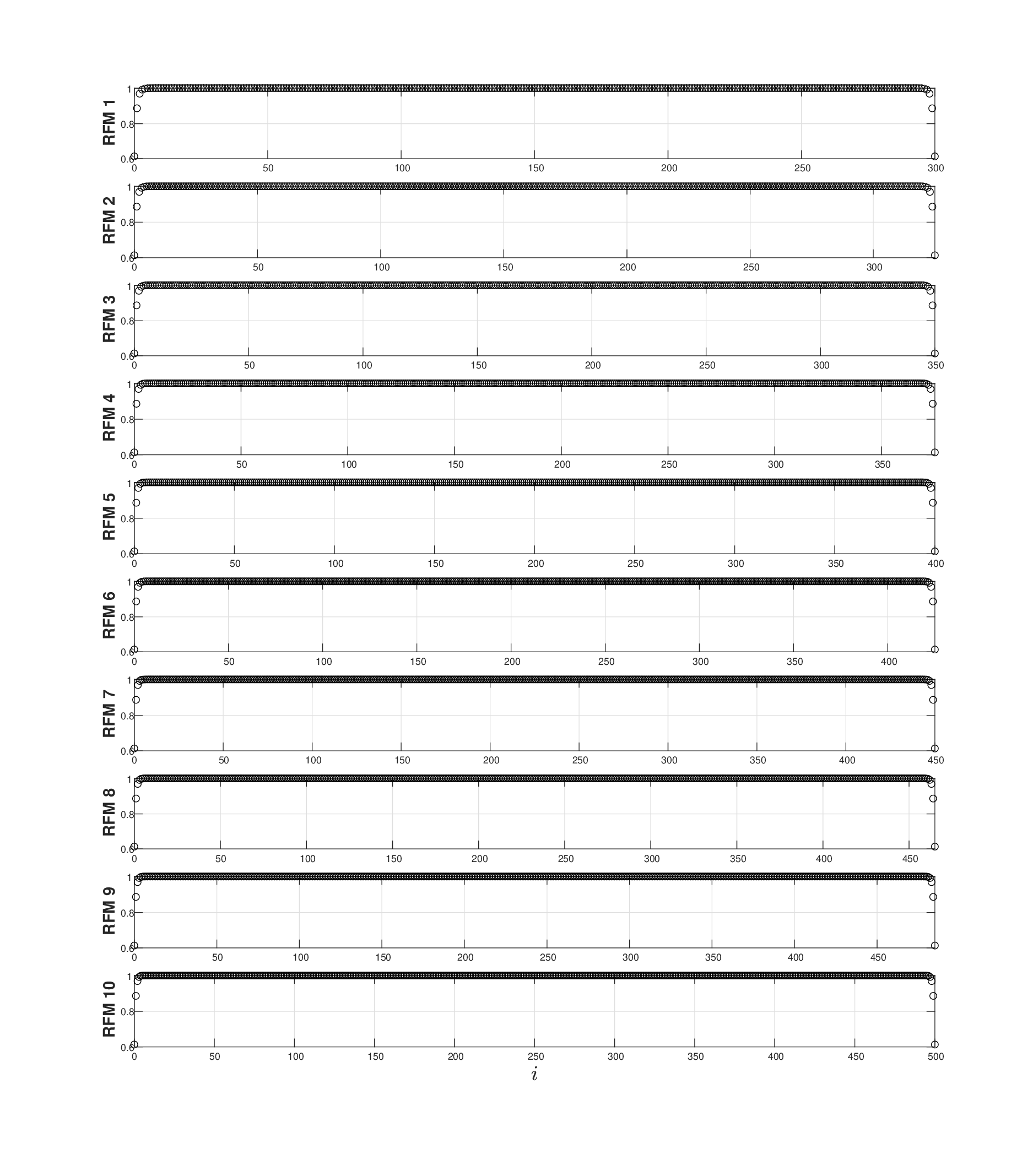}
  \caption{Optimal rates as a function of the location~$i$ along the RFM 
  in a network of $m=10$
  RFMs with  the lengths in~\eqref{eq:leng}.}
\label{fig:plit10}
\end{center}
\end{figure}

\section{Discussion}

Turnpike structures in optimization problems
were first identified in mathematical economics~\cite{Rams28,DoSS58} and later   in optimal control theory, see, for example,~\cite{FauG22,TRELAT201581} and the references therein.
More recently, Kaminer et al.~\cite{kaminer2026turnpikepropertyeigenvalueoptimization}
showed that the optimal allocation of transition rates in a single~RFM
exhibits a turnpike structure. 

The present work extends this concept from a single transcript to a network of competing transcripts sharing a common resource budget. We prove that the optimal solution exhibits a \emph{multi-turnpike structure}, where 
each transcript independently admits  a turnpike profile while the shared resources are optimally distributed across the entire network. The proof of this result uses a 
 hierarchical  optimization  argument showing that inside the 
large-scale networked optimization problem hide
local solutions that must  also be optimal. 
In summary, 
global competition for shared translational resources induces   local turnpike architectures within each transcript.

Our results suggest 
that global translational regulation in the cell  may be implemented through relatively simple molecular mechanisms, primarily   tuning initiation and termination regions of coding sequences, rather than relying on complex sensing and signaling pathways that   continuously regulate the 
elongation rates at individual codon positions along each transcript.  
Furthermore,
 our results suggest that the ubiquitous codon ramps observed across organisms  are 
 a universal consequence of optimal allocation of shared translational resources, providing also a potential explanation for conserved features of translation across species.

More generally, the  multi-turnpike structure may reflect a general principle in biological resource allocation: when optimizing a global objective under a shared constraint \cite{Sabi2019}, optimal strategies tend to equalize marginal returns across most components of the system, with heterogeneity confined to boundaries or interfaces. This principle may apply broadly to other biological
flow-based processes such as transcription \cite{Cohen2018,Edri2014}. 

Finally, our findings may provide practical design principles for synthetic biology. To maximize protein production in intracellular circuits regulated at the translational level, engineering efforts should concentrate on the initiation and termination regions of coding sequences, where modest changes in translation rates can substantially affect protein production. In contrast, when no competing constraints, such as overlapping coding or regulatory functions, are imposed, the interior of the coding sequence should be designed to maintain high and nearly uniform translation rates.

\section{Supplementary Material}
We detail here the proofs of the main results,
and present 
more mathematical definitions and derivations. 

\subsection{Ribosome Flow Model}
The RFM is a nonlinear dynamical compartmental model for the flow of particles (e.g., ribosomes) along an ordered chain of~$n$ sites (e.g., the mRNA coarse-grained into a set of~$n$ consecutive groups of codons). Let~$x_i(t) $ denote the density at site~$i$ at time~$t$. This is normalized such that~$x_i(t)\in[0,1]$ for all~$t$, with~$x_i(t)=0$ implies that site~$i $ is completely free,  and~$x_i(t)=1$ implies that it is completely full. The~RFM also includes~$n+1$ positive parameters~$\lambda_0,\dots,\lambda_n$, where~$\lambda_i$ is the transition rate from site~$i-1$ to site~$i$. In particular, ~$\lambda_0$ is the entry rate to the first site, and~$\lambda_n$ is the exit rate from the last site. 

The dynamic equations of the RFM are
\be\label{eq:rfm_dynam}
\dot x_i  = \lambda_{i-1}x_{i-1}(1-x_i)- \lambda_{i }x_{i }(1-x_{i+1}),\quad i=1,\dots,n.
\ee
Here~$x_0:=1$ and~$x_{n+1}:=0$. The term~$\lambda_{i-1}x_{i-1}(1-x_i)$ describes the flow of particles from site~$i-1$ to site~$i$.
This is proportional to the  rate~$\lambda_{i-1}$, to the density~$x_{i-1}$  of particles in site~$i-1$, and to~$(1-x_i)$, that is, the free space in site~$i$. 
Similarly,  the term~$\lambda_{i }x_{i }(1-x_{i+1})$
 describes the flow of particles from site~$i $ to site~$i+1$, so
 \eqref{eq:rfm_dynam} is just a balance equation stating that the change in the density in a site is the flow into the site minus the flow out of this site.

The nonlinearity of the~RFM arises from the products of state variables in~\eqref{eq:rfm_dynam}, making it a nonlinear system of first-order ordinary differential equations.

 Let~$R(t):=\lambda_n x_n(t)$ denote the flow of particles out of the last site. This is the protein production rate at time~$t$. 
Let~$\lambda:=\begin{bmatrix}
    \lambda_0&\dots&\lambda_n
\end{bmatrix}^\top$ denote the set of positive transition rates. An important property of the RFM is that it admits a unique equilibrium point~$e\in(0,1)^n$, which depends only on the rates i.e.~$e=e(\lambda)$, and for any initial condition~$x(0)\in[0,1]^n $
the solution of the RFM converges to~$e$~\cite{margaliot2012stability}. At the steady-state, that is, when~$x_i=e_i$ for all~$i$, the flow into each site and out of each site are equal, so the densities remain constant.  The steady state does not change if the vector $\lambda$ is multiplied by a scalar $\rho>0$, i.e., $e(\rho\lambda)=e(\lambda)$ for all $\rho>0$. As a consequence, the protein production rate in the steady state scales linearly with $\rho>0$, which implies that the maximal production rate scales linearly with the bound on $\sum_{i=1}^n\lambda_i$. 

Since the convergence to~$e$ is (almost) exponential \cite{Margaliot2016ContractionSmallTransients}, the ribosome densities rapidly approach their steady-state values. Accordingly, after a relatively short transient, the protein production rate is well approximated by its steady-state value,
$R:=\lambda_n e_n$.

\subsection{Utility Function}\label{sub:add_cond_utility_func}
We describe the additional technical condition that the utility function~$F=F(R_1,R_2)$  must satisfy.  
To formulate it, we first introduce some notation. 
Let~$\bar R^s(n)$ denote the maximum of the steady state production rate~$R$ in an~RFM with length~$n$ and  with the constraints~$\lambda_i\geq0 $ for all~$i$ and~$\sum_{i=0}^n \lambda_i \leq n+1 $.
Let~$\bar \lambda^s(n) \in (0, n+1)^{n+1}$ denote the corresponding unique maximizer (which has the turnpike property~\cite{kaminer2026turnpikepropertyeigenvalueoptimization}), 
and let
\be\label{eq:small_r}
\bar r^s(n):=\frac{\bar R^s(n)}{n+1},
\ee
that is, the optimal production rate in the single RFM  scaled by the total budget (see Fig.~\ref{fig:bigsmall_R}).  
The superscript~$s$ denotes that these are  all optimal values for the problem of maximizing the production rate in a 
\emph{single}~RFM.

 \begin{figure}[t]
\begin{center}
\includegraphics[scale=0.8]{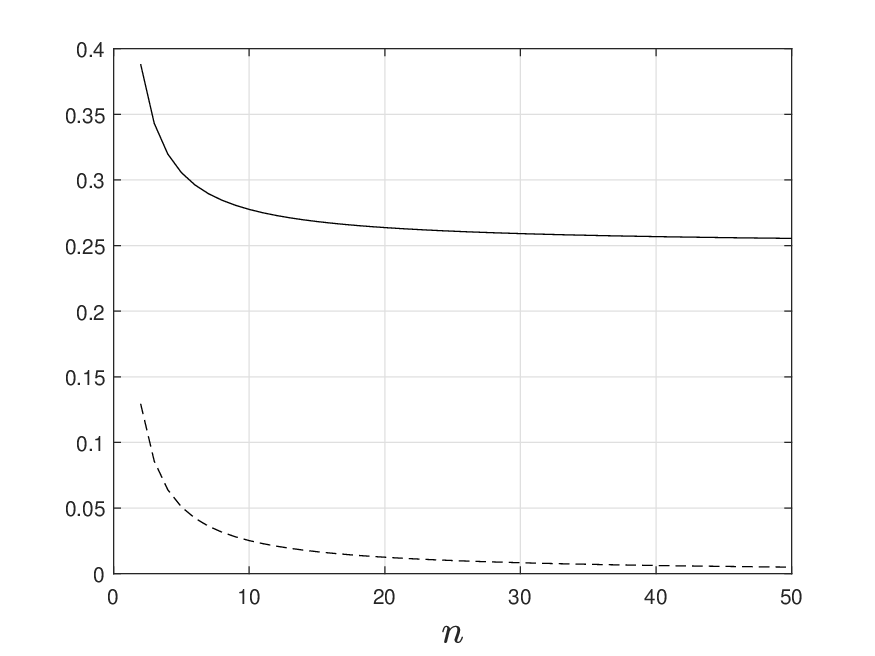}
  \caption{The functions~$\bar R^s(n)$ (solid line) and~$\bar r^s(n)$ (dashed line)
  for~$n=2,3,\dots,50$.  }
\label{fig:bigsmall_R}
\end{center}
\end{figure}

Let
\[
N:=n_1+n_2+2
\]
denote the total budget.  
Define the solid triangle:  
\[
\Delta_{n_1, n_2} := \{R \in [0, \infty)^2 
\suchthat 
 \frac{R_1}{\bar r^s (n_1)} + \frac{R_2}{\bar  r^s ( n_2)} \leq N\} . 
\]
The additional technical condition on the utility function~$F$ is that all maximizers of~$F$ restricted to the triangle $\Delta_{n_1, n_2}$ are located on the upper boundary $\{R \in [0, \infty)^2 \suchthat  \frac{R_1}{\bar r^s(n_1)} + \frac{R_2}{\bar r^s(n_2)} =N \}$.
Intuitively speaking, 
this means that  an 
optimal solution always uses all the cellular budget~$N$. 
If~$F(R_1,R_2)$ is strictly increasing  in~$R_1,R_2$ then this condition  automatically holds. This implies that 
the functions~$F_1$ and~$F_2$ in Section~\ref{sec:multi}
satisfy this condition. The functions~$F_3,F_4$ also satisfy this condition.

\subsection{Proof of Theorem~\ref{thm:main}}

The next result characterizes the optimal solutions of  Problem~\ref{prob:max_global}.

\begin{Theorem}\label{thm:variant_of_main}
Assume that~$F(R_1,R_2)$ is continuous, increasing in~$R_1$ and~$R_2$, and satisfies the additional technical condition. 
    Let~$(\bar\lambda,\bar\eta)$
be an optimal solution of Problem~\ref{prob:max_global}. Denote~$\bar s_1:=\sum_{i=0}^{n_1}\bar\lambda_i$ and~$\bar s_2:=\sum_{j=0}^{n_2}\bar\eta_j$.
Then 
\[
\bar\lambda =  \frac{\bar s_1}{n_1+1} \bar \lambda^s ( n_1)   \text{ and } 
\bar\eta =  \frac{\bar s_2}{n_2+1} \bar \lambda^s ( n_2 ).
\]
\end{Theorem}
In other words,~$\bar \lambda$ and~$\bar \eta$ are both scaled versions of solutions to the problem of maximizing the production rate in a single RFM, and are thus in particular  turnpike solutions. This  proves   in particular 
Theorem~\ref{thm:main}. 

\begin{proof}[Proof of Theorem~\ref{thm:variant_of_main}]
It follows from the results in~\cite{kaminer2026turnpikepropertyeigenvalueoptimization} that 
\begin{align*}
0&\leq R_1(\bar\lambda) \leq \frac{\bar s_1}{n_1+1} \bar R^s(n_1)
= \bar s_1 \bar r^s ( n_1 ) , \\
0&\leq R_2(\bar\eta) \leq \frac{\bar s_2}{n_2+1} \bar R^s ( n_2)  =
\bar s_2 \bar r^s ( n_2 ) .
\end{align*}
This implies $(R_1(\bar\lambda),R_2(\bar \eta))\in \Delta_{n_1, n_2}$. Using the additional condition on the utility function, it now follows that $(R_1(\bar\lambda)/\bar r^s(n_1)) +(R_2(\bar\eta)/\bar r^s (n_2))=N$, and therefore $R_1(\bar\lambda)=\bar s_1 \bar r^s(n_1)$ and $R_2(\bar\eta)=\bar s_2 \bar r^s (n_2)$. The uniqueness result for the maximizer in a single RFM  in~\cite{kaminer2026turnpikepropertyeigenvalueoptimization} completes the proof.
\end{proof}

\subsection{Budget Allocation Between Transcripts}
The value~$\bar s_i$ is the total budget allocated to RFM number~$i$  in the solution of the optimal control problem. 
 To further analyze this allocation,
 define a function~$f:[0,1]\to \R_{\geq 0}$ by
\be\label{eq:def_small_f}
f(q):=
{F}\left(q N \bar r^s(n_1),
 (1-q) N \bar r^s(n_2)\right ).
\ee
Theorem~\ref{thm:variant_of_main} implies that the problem of maximizing 
$f(q)$ over~$q\in[0,1]$
is equivalent to Problem~\ref{prob:max_global}. A 
 maximizer~$\bar q$ of this problem is in one-to-one correspondence  with a  maximizer 
 of Problem~\ref{prob:max_global} with
 an optimal  budget allocation of~$\bar qN$ [$(1-\bar q)N$] to the first [second] RFM,
  and the maximizing transition rates  are  
\[
(\frac{\bar q N }{n_1+1} \bar \lambda^s(n_1), \frac{(1-\bar q )N}{n_2+1} \bar\lambda^s(n_2) ) 
\]
 yielding the maximal utility~$f(\bar q)$.

 With this information, we can analyze the budget allocation to each RFM for the  utility functions introduced above. For each function~$F_i$ we study the equivalent function~$f_i$ defined by~\eqref{eq:def_small_f}. 
 
 For the function~$F_1$, we get 
 \begin{align*}
f_1(q ) 
& = w_2N\bar r^s(n_2)+ qN( w_1 \bar r^s(n_1)-w_2\bar r^s(n_2) ).
\end{align*}
If~$w_1 \bar r^s(n_1)=w_2\bar r^s(n_2)$
then~$f_1$  is independent of~$q$ implying that any allocation of the total budget~$N$ between the two RFMs is optimal. This is the case for example when~$w_1=w_2$ and~$n_1=n_2$, that is, the utility function gives the same weighting to the two production rates and the two RFMs have equal length.
If~$w_1 \bar r^s(n_1)  > w_2\bar r^s(n_2)$
[$w_1 \bar r^s(n_1)  < w_2\bar r^s(n_2)$]
then~$f_1(q)$ admits a unique optimizer~$\bar q=1$ [$\bar q=0$], that is, the optimal solution is always to allocate all the budget~$N$ to one of the~RFMs. 

If~$w_1=w_2$ we can deduce from the strict monotonicity of $\bar r^s(n)$ displayed in Fig.~\ref{fig:bigsmall_R} that the optimal allocation is to give all the budget to the shorter~RFM and zero budget to the longer one. So far, we do not have a mathematical proof for the strict monotone decrease of $\bar r^s(n)$ with $n$. However, the 
results in~\cite{kaminer2026turnpikepropertyeigenvalueoptimization} imply that 
$\bar r^s(n)=\bar R^s(n)/(n+1)$ satisfies
\be
\frac{1}{4n}<
\bar r^s(n) <  
\frac{1}{4\left(n+1-4\ln(2)\right )}.
\ee  
These bounds imply that if~$n_1  \geq   n_2+2$ then~$\bar r^s(n_1)<\bar r_s(n_2)$  and the unique maximizer allocates all the budget to the shorter~RFM.

For the function~$F_2$, we get 
\begin{align*}
f_2( q )   =w_1 \ln\big (1+ qN\bar r^s(n_1) \big )+w_2\ln \big (1+(1-q) N\bar r^s(n_2)\big ).
\end{align*}
This function is strictly concave in~$q$, so the optimization problem admits a unique solution~$\bar q\in[0,1]$. Calculating~$\frac{\partial}{\partial q}f_2(q)$  and comparing the derivative to zero gives that
\[
 \tilde q 
= \frac{w_1\bar r^s(n_1) ( 
1+N\bar r^s(n_2)) -w_2 \bar r^s(n_2)
 }
{(w_1+w_2)N\bar r^s(n_1)\bar r^s(n_2)}.
\]
 If~$\tilde q$ lies in the interval $[0,1]$ then the maximizer is given by $\bar q = \tilde q$. If~$\tilde q <0 \; [\tilde q > 1]$ we have~$\bar q = 0 \;[\bar q = 1]$.

For the function~$F_3$, we get 

\begin{align*}
f_3(q ) 
=q(1-q) N^2\bar r^s(n_1)\bar r^s(n_2) , 
\end{align*}
and the unique maximizer is~$\bar q=1/2$. Thus, the optimal solution is always to split the total budget equally between the two~RFMs (see the caption in Figure~\ref{fig:5040}).

For the function~$F_4$, we get
\begin{align*}
f_4( q) =N\min\{   q \bar r^s(n_1),(1-q) \bar r^s(n_2)  \},
\end{align*}
and this admits a unique maximizer 
\[
\bar q= \frac{\bar r^s(n_2)} { \bar r^s(n_1)+\bar r^s(n_2) } .
\]

\section{Appendix A: Extension to the  case of~$m$ transcripts}\label{TAppendix}
For the sake of simplicity,    we stated all the results above for the particular case of $m=2$ RFMs. We now explain how these results carry over to the case with a  general number~$m$ of RFMs. 

Let~$n_k>1$ be the length of the~$k$th RFM, and denote its vector of
non-negative
transition rates by~$\lambda^k=\begin{bmatrix}
    \lambda_0^k&\lambda_1^k&\dots&\lambda_{n_k}^k \end{bmatrix}$, and its steady state production rate by~$R_k(\lambda^k)$. The utility function is~$F(R_1(\lambda^1),\dots,R_m(\lambda^m))$, and
    the optimization problem is
\[
\max F(R_1(\lambda^1),\dots,R_m(\lambda^m))
\]
subject to~$\lambda^k_i\geq 0$ for all~$i,k$ and
\[
\sum_{k=1}^m 
\sum_{i=0}^{n_k} \lambda^k_i\leq \sum_{k=1}^m (n_k+1).
 \]

We assume that the  utility function~$F:\R^m_{\geq 0} \to \R$ is continuous, increasing in every~$R_i$, and satisfies an additional technical condition as follows. 

Let
\[
N:=\sum_{k=1}^m(n_k+1)
\]
denote the total budget.  
Define the solid triangle:  
\[
\Delta_{n_1, \dots,n_m} := \{R \in [0, \infty)^m 
\suchthat \sum_{k=1}^m \frac{R_k}{\bar r^s(n_k)} \leq N\} , 
\]
where~$\bar r^s(n)$ is defined in~\eqref{eq:small_r}.

The additional technical condition on the utility function~$F$ is that all maximizers of~$F$ restricted to this  are located on the upper boundary $
\Delta_{n_1, \dots,n_m} := \{R \in [0, \infty)^m 
\suchthat \sum_{k=1}^m \frac{R_k}{\bar r^s(n_k)} = N\} $. 
Intuitively speaking, 
this means that an
optimal solution always uses all the cellular budget~$N$. 
If~$F(R_1,\dots,R_m)$ is strictly increasing  in every~$R_i$ then this condition  automatically holds. 

Utility functions that satisfy all the required conditions include the following
\begin{itemize}
\item 
$F_1(R_1(\lambda^1),\dots,R_m(\lambda^m)):=
\sum_{k=1}^m w_k  R_k(\lambda^k)$,
where the~$w_k$s are positive weights.   
\item $F_2(R_1(\lambda^1),\dots,R_m(\lambda^m)) : =\sum_{k=1}^m w_k \ln(1+R_k(\lambda^k)) $, 
where the~$w_k$s are positive weights.
\item 
$F_3(R_1(\lambda^1),\dots,R_m(\lambda^m))
:=\prod_{k=1}^m R_k(\lambda^k)$. 
\item $F_4(R_1(\lambda^1),\dots,R_m(\lambda^m)) : =\min\{R_1(\lambda^1),
\dots,R_m(\lambda^m)
\}$. 
\end{itemize}

Since~$F$ is continuous 
and the feasible set of parameters is closed and bounded,  
an optimal solution~$ \bar \lambda^1,\dots,\bar \lambda^m$ 
exists.
The next result shows in particular
that every solution of the optimization problem with~$m$ RFMs 
admits a multi-turnpike structure.

\begin{Theorem}\label{thm:m_rfm_s_main}
Assume that~$F$ is continuous, increasing in every~$R_k$, and satisfies the additional technical condition. 
    Let~$ \bar\lambda^1,\dots,\bar\lambda^m $
be an optimal solution. Denote~$\bar s_k:=\sum_{i=0}^{n_k}\bar\lambda^k_i$,
$k=1,\dots,m$. 
Then 
\[
\bar\lambda^k =  \frac{\bar s_k}{n_k+1} \bar \lambda^s ( n_k)  ,\quad k=1,\dots,m.
\]
\end{Theorem}
In other words,
every~$\bar \lambda^k$
is a scaled version of a solution to the problem of maximizing the production rate in a single RFM, and is thus a turnpike solution. 

The proof of this result is very similar to the proof of Theorem~\ref{thm:variant_of_main} and is thus omitted. 

Note that 
Theorem~\ref{thm:m_rfm_s_main} allows to reduce the problem of optimizing $F(R_1(\lambda^1),\dots,R_m(\lambda^m))$ to the simpler task of
determining the maximizers~$\bar q$ of the function $f(q):={F}\left(q_1 N \bar r^s(n_1), \ldots, 
 q_m N \bar r^s(n_m)\right )$ on the $(m-1)$-dimensional polytope $\Delta := \{q \in [0, 1]^m 
\suchthat \sum_{k=1}^m q_k = 1\}$. For each maximizer~$\bar q$  the corresponding
optimal rates are then given by 
\[
\bar\lambda^k =  \frac{\bar q_k N}{n_k+1} \bar \lambda^s ( n_k), \quad k=1,\dots,m.
\]

\section*{
Acknowledgments and funding sources }
 The research of~MM is partly supported by a  research grant from the Israeli Science Foundation~(ISF). The research of MM, TK, and LG was supported by DFG Grant No.\ 470999742.    
The research of~TT is partially supported by a research grant from the Israeli  
Ministry of Innovation, Science and Technology~(MOST) and  a HORIZON-EIC-2025 Pathfinder grant.

\section*{Conflict of interest statement}
The authors declare that they have no competing interests.

\section*{Data sharing}
All numerical results and figures were generated using MATLAB routines developed for this study. The~MATLAB code   is available from the corresponding author upon  request.

\section*{Disclosure of Delegation to Generative AI}
The authors declare the use of generative AI in the research and writing process. According to the GAIDeT taxonomy (2025), the following tasks were delegated to GAI tools under full human supervision:

- Proofreading and editing

\noindent The GAI tool used was: ChatGpt.
Responsibility for the final manuscript lies entirely with the authors.


\begin{thebibliography}{10}
\providecommand{\url}[1]{#1}
\csname url@samestyle\endcsname
\providecommand{\newblock}{\relax}
\providecommand{\bibinfo}[2]{#2}
\providecommand{\BIBentrySTDinterwordspacing}{\spaceskip=0pt\relax}
\providecommand{\BIBentryALTinterwordstretchfactor}{4}
\providecommand{\BIBentryALTinterwordspacing}{\spaceskip=\fontdimen2\font plus
\BIBentryALTinterwordstretchfactor\fontdimen3\font minus \fontdimen4\font\relax}
\providecommand{\BIBforeignlanguage}[2]{{%
\expandafter\ifx\csname l@#1\endcsname\relax
\typeout{** WARNING: IEEEtranS.bst: No hyphenation pattern has been}%
\typeout{** loaded for the language `#1'. Using the pattern for}%
\typeout{** the default language instead.}%
\else
\language=\csname l@#1\endcsname
\fi
#2}}
\providecommand{\BIBdecl}{\relax}
\BIBdecl

\bibitem{Bahiri2021}
S.~Bahiri~Elitzur, R.~Cohen-Kupiec, D.~Yacobi, L.~Fine, B.~Apt, A.~Diament, and T.~Tuller, ``Prokaryotic {rRNA-mRNA} interactions are involved in all translation steps and shape bacterial transcripts,'' \emph{RNA Biol.}, vol.~18, no. sup2, pp. 684--698, 2021.

\bibitem{total_mrna_PNAS2024}
L.~Calabrese, L.~Ciandrini, and M.~C. Lagomarsino, ``How total {mRNA} influences cell growth,'' \emph{Proceedings of the National Academy of Sciences}, vol. 121, no.~21, p. e2400679121, 2024.

\bibitem{Cohen2018}
E.~Cohen, Z.~Zafrir, and T.~Tuller, ``A code for transcription elongation speed,'' \emph{RNA Biol.}, vol.~15, no.~1, pp. 81--94, 2018.

\bibitem{DoSS58}
R.~Dorfman, P.~A. Samuelson, and R.~M. Solow, \emph{Linear programming and economic analysis}, ser. A Rand Corporation Research Study.\hskip 1em plus 0.5em minus 0.4em\relax New York-Toronto-London: McGraw-Hill, 1958.

\bibitem{Edri2014}
S.~Edri, E.~Gazit, E.~Cohen, and T.~Tuller, ``The {RNA} polymerase flow model of gene transcription,'' \emph{IEEE Trans Biomed Circuits Syst.}, vol.~8, no.~1, pp. 54--64, 2014.

\bibitem{RFM_NEGATIVE_FEEDBACK}
A.~Ehrman, T.~Kriecherbauer, L.~Gr\"une, and M.~Margaliot, ``Negative feedback and oscillations in a model for {mRNA} translation,'' \emph{J. Royal Society Interface}, vol.~22, p. 20250338, 2025.

\bibitem{FauG22}
T.~Faulwasser and L.~Gr\"une, ``Turnpike properties in optimal control,'' in \emph{Numerical Control: Part A}, E.~Tr\'elat and E.~Zuazua, Eds.\hskip 1em plus 0.5em minus 0.4em\relax Elsevier, 2022, vol.~24, pp. 367--400.

\bibitem{pilpel2017}
I.~Frumkin, D.~Schirman, A.~Rotman, F.~Li, L.~Zahavi, E.~Mordret, O.~Asraf, S.~Wu, S.~F. Levy, and Y.~Pilpel, ``Gene architectures that minimize cost of gene expression,'' \emph{Molecular Cell}, vol.~65, no.~1, pp. 142--153, 2017.

\bibitem{num_mrna_2014}
S.~Islam, A.~Zeisel, S.~Joost, G.~L. Manno, P.~Zajac, M.~Kasper, P.~Lönnerberg, and S.~Linnarsso, ``Quantitative single-cell {RNA}-seq with unique molecular identifiers,'' \emph{Nature Methods}, vol.~2, pp. 163--166, 2014.

\bibitem{aditi_networks}
A.~Jain, M.~Margaliot, and A.~K. Gupta, ``Large-scale {mRNA} translation and the intricate effects of competition for the finite pool of ribosomes,'' \emph{J. R. Soc. Interface}, vol.~19, p. 2022.0033, 2022.

\bibitem{Jia_2026}
\BIBentryALTinterwordspacing
L.~Jia, Y.~Mao, S.~Uematsu, X.~A. Liu, L.~Dong, L.~H. França~de Lima, and S.-B. Qian, ``Profiling of terminating ribosomes reveals translational control at stop codons,'' 2026. [Online]. Available: \url{http://dx.doi.org/10.7554/eLife.109257.2}
\BIBentrySTDinterwordspacing

\bibitem{kaminer2026turnpikepropertyeigenvalueoptimization}
\BIBentryALTinterwordspacing
A.~Kaminer, T.~Kriecherbauer, L.~Grüne, and M.~Margaliot, ``A turnpike property in an eigenvalue optimization problem,'' 2026, submitted. [Online]. Available: \url{https://arxiv.org/abs/2601.13756}
\BIBentrySTDinterwordspacing

\bibitem{fierce_compete}
R.~Katz, E.~Attias, T.~Tuller, and M.~Margaliot, ``Translation in the cell under fierce competition for shared resources: a mathematical model,'' \emph{J. Royal Society Interface}, vol.~19, p. 20220535, 2022.

\bibitem{Margaliot2016ContractionSmallTransients}
M.~Margaliot, E.~D. Sontag, and T.~Tuller, ``Contraction after small transients,'' \emph{Automatica}, vol.~67, pp. 178--184, 2016.

\bibitem{margaliot2012stability}
M.~Margaliot and T.~Tuller, ``Stability analysis of the ribosome flow model,'' \emph{IEEE/ACM Trans. Computational Biology and Bioinformatics}, vol.~9, no.~5, pp. 1545--1552, 2012.

\bibitem{Peeri2020}
M.~Peeri and T.~Tuller, ``High-resolution modeling of the selection on local {mRNA} folding strength in coding sequences across the tree of life,'' \emph{Genome Biol.}, vol.~21, no.~1, p.~63, 2020.

\bibitem{rfm_max}
G.~Poker, Y.~Zarai, M.~Margaliot, and T.~Tuller, ``Maximizing protein translation rate in the nonhomogeneous ribosome flow model: A convex optimization approach,'' \emph{J. Royal Society Interface}, vol.~11, no. 100, p. 20140713, 2014.

\bibitem{rfm_sense}
G.~Poker, M.~Margaliot, and T.~Tuller, ``Sensitivity of {mRNA} translation,'' \emph{Sci. Rep.}, vol.~5, no. 12795, 2015.

\bibitem{Rams28}
F.~P. Ramsey, ``A mathematical theory of saving,'' \emph{The Economic Journal}, vol.~38, no. 152, pp. 543--559, 1928.

\bibitem{Raveh2016}
A.~Raveh, M.~Margaliot, E.~Sontag, and T.~Tuller, ``A model for competition for ribosomes in the cell,'' \emph{J. Royal Society Interface}, vol.~13, no. 116, p. 20151062, 2016.

\bibitem{reuveni2011genome}
S.~Reuveni, I.~Meilijson, M.~Kupiec, E.~Ruppin, and T.~Tuller, ``Genome-scale analysis of translation elongation with a ribosome flow model,'' \emph{PLOS Computational Biology}, vol.~7, no.~9, p. e1002127, 2011.

\bibitem{Roux2012}
P.~P. Roux and I.~Topisirovic, ``Regulation of {mRNA} translation by signaling pathways,'' \emph{Cold Spring Harb Perspect Biol.}, vol.~4, no.~11, p. a012252, 2012.

\bibitem{Sabi2019}
R.~Sabi and T.~Tuller, ``Modelling and measuring intracellular competition for finite resources during gene expression,'' \emph{J. R. Soc. Interface}, vol.~16, no. 154, p. 20180887, 2019.

\bibitem{resource-aware2024}
K.~Sechkar, H.~Steel, G.~Perrino, and G.-B. Stan, ``A coarse-grained bacterial cell model for resource-aware analysis and design of synthetic gene circuits,'' \emph{Nature Comm.}, vol.~15, p. 1981, 2024.

\bibitem{TEODOROWICZ2026169765}
W.~Teodorowicz and O.~Mühlemann, ``Molecular determinants and therapeutic targeting of stop codon readthrough in eukaryotic translation,'' \emph{J. Molecular Biology}, p. 169765, 2026.

\bibitem{TEUSINK2026103391}
B.~Teusink, P.~Grigaitis, M.~Remeijer, F.~Bruggeman, and R.~Steuer, ``Resource allocation models: theory and applications in microbial biotechnology,'' \emph{Current Opinion in Biotechnology}, vol.~97, p. 103391, 2026.

\bibitem{TRELAT201581}
E.~Trélat and E.~Zuazua, ``The turnpike property in finite-dimensional nonlinear optimal control,'' \emph{J. Diff. Eqns.}, vol. 258, no.~1, pp. 81--114, 2015.

\bibitem{TULLER2010344}
T.~Tuller, A.~Carmi, K.~Vestsigian, S.~Navon, Y.~Dorfan, J.~Zaborske, T.~Pan, O.~Dahan, I.~Furman, and Y.~Pilpel, ``An evolutionarily conserved mechanism for controlling the efficiency of protein translation,'' \emph{Cell}, vol. 141, no.~2, pp. 344--354, 2010.

\bibitem{TULLER2011}
T.~Tuller, I.~Veksler-Lublinsky, N.~Gazit, M.~Kupiec, E.~Ruppin, and M.~Ziv-Ukelson, ``Composite effects of gene determinants on the translation speed and density of ribosomes,'' \emph{Genome Biol.}, vol.~12, no.~11, p. R110, 2011.

\bibitem{NAR2015}
T.~Tuller and H.~Zur, ``Multiple roles of the coding sequence 5' end in gene expression regulation,'' \emph{Nucleic Acids Res.}, vol.~43, no.~1, pp. 13--28, 2015.

\bibitem{min_spring}
Y.~Zarai and M.~Margaliot, ``On minimizing the maximal characteristic frequency of a linear chain,'' \emph{IEEE Trans. Autom. Control}, vol.~62, no.~9, pp. 4827--4833, 2017.

\bibitem{Zeng2021CellularResourceAllocation}
H.~Zeng, R.~Rohani, W.~E. Huang, and A.~Yang, ``Understanding and mathematical modelling of cellular resource allocation in microorganisms: A comparative synthesis,'' \emph{BMC Bioinformatics}, vol.~22, p. 467, 2021.

\end{thebibliography}
\end{document}